\documentclass[journal]{IEEEtran}
\usepackage{amssymb}
\usepackage{amsmath,amssymb,amsfonts}
\usepackage{multirow}
\usepackage{graphicx}
\usepackage{textcomp}
\usepackage{amsthm}
\renewcommand{\vec}[1]{\boldsymbol{#1}}
\usepackage{setspace}

\usepackage{graphicx}
\usepackage{algorithm}
\usepackage{algorithmic}
\usepackage{amsmath}
\usepackage{color}
\usepackage{amsfonts}

\ifCLASSINFOpdf
\else
\fi

\newtheorem{thm}{Theorem}
\newtheorem{lem}{Lemma}

\newtheorem{defn}{Definition}

\newtheorem{exmp}{Example}
\newtheorem{rem}{Remark}
\newtheorem{pro}{Problem}

\begin{document}
%
\title{A k-hop Collaborate Game Model: Adaptive Strategy to Maximize Total Revenue}
%
%
%

\author{Jianxiong Guo,
	Weili Wu~\IEEEmembership{Member,~IEEE}
	\IEEEcompsocitemizethanks{\IEEEcompsocthanksitem J. Guo, and W. Wu are with the Department of Computer Science, Erik Jonsson School of Engineering and Computer Science, Univerity of Texas at Dallas, Richardson, TX, 75080 USA\protect\\
		E-mail: jianxiong.guo@utdallas.edu
	}
	\thanks{Manuscript received April 19, 2005; revised August 26, 2015.}}

%
%

\markboth{Journal of \LaTeX\ Class Files,~Vol.~14, No.~8, August~2015}%
{Shell \MakeLowercase{\textit{et al.}}: Bare Demo of IEEEtran.cls for IEEE Journals}
%



\maketitle

\begin{abstract}
In Online Social Networks (OSNs), interpersonal communication and information sharing are happening all the time, and it is real-time. When a user initiates an activity in OSNs, immediately, he/she will have a certain influence in his/her friendship circle naturally, some users in the initiator's friendship circle will be attracted to participate in this activity. Based on such a fact, we design a k-hop Collaborate Game Model, which means that an activity initiated by a user can only influence those users whose distance are within k-hop from this initiator in OSNs. Besides, we introduce the problem of Revenue Maximization under k-hop Collaborate Game (RMKCG), which identifies a limited number of initiators in order to obtain revenue as much as possible. Collaborate Game Model describes in detail how to quantify revenue and the logic behind it. We do not know how many followers would be generated from an activity in advance, thus, we need to adopt an adaptive strategy, where the decision who is the next potential initiator depends on the results of past decisions. Adaptive RMKCG problem can be considered as a new stochastic optimization problem, and we prove it is NP-hard, adaptive monotone, but not adaptive submodular. But in some special cases, it is adaptive submodular and an adaptive greedy strategy can obtain a $(1-1/e)$-approximation by adaptive submodularity theory. Due to the complexity of our model, it is hard to compute the marginal gain for each candidate user, then, we propose a convenient and efficient computational method. The effectiveness and correctness of our algorithms are validated by heavy simulation on real-world social networks eventually.
\end{abstract}


%
\IEEEpeerreviewmaketitle

\section{Introduction}
%
%
%
%

\IEEEPARstart{T}{he} online social platforms were developing quickly in the last decades and derived a series of famous technological companies, such as Facebook, Twitter, LinkedIn and Tencent. There are billions of people sharing their emotions and discussing current affairs in these platforms. There are more than $1.52$ billion users active daily on Facebook and $321$ million users active monthly on Twitter. The users' friendship in these social platforms can be described by an online social network (OSNs), which is an undirected graph, including individuals and their relationship. Domingos and Richardson \cite{domingos2001mining} \cite{richardson2002mining} was the first one to propose the concept of "Viral Marketing", which aims to attract follow-ups as many as possible by giving free or coupon samples to the most influential users in OSNs. This groundbreaking researches has had a profound impact on later generations. Inspired by that, Influence Maximization was proposed to model the spread of trust, advertisements or innovations abstractly. Kempe et al. \cite{kempe2003maximizing} regarded it as a combinatorial optimization, which aims to select a small subset of users such that the expected number of follow-up adoptions can be maximizied under the size constraint. Besides, they propsoed two classical diffusion models called Independent Cascade model and Linear Threshold model, and prove it is NP-hard, monotone and submodualr under these two models. Since this seminal work, a series of variant problems and models used for different scenarios and constraints \cite{chen2011influence} \cite{zhang2014recent} \cite{8850214} \cite{guo2019novel} appeared constantly, such as profit maximization, rumor blocking and adaptive submodular problem. And some researchers focus on the extension \cite{chen2016robust} \cite{leskovec2007cost} \cite{wang2017activity} \cite{zhang2016profit} \cite{yuan2017active} \cite{8952599} from a practical viewpoint.

However, most existing researches about maximization problem, regardless of influence or profit, all base on counting the number of single user that follows our "influence". This model is indeed effective and valid in most cases, but it does not cover all scenarios. Considering some user in a social platform, such as Facebook, is invited by some organizations to launch an activity, after lauching, his/her friends, or friends of friends, may be interested in this activity and choose to participate in it. At this time, our influence or profit depends on the benefits from activity, which is related to the number of people involved in this event. Based on that, we propose Collaborate Game Model. Considering a game company, they plan to promote their multiplayer game in some social platform by inviting some users to play this game. Once these users accept this invitation, they initate this game, called initiator, and he/she will attract friends within a certain range to participate in it. In our model, we have a k-hop assumption that the maximum influence range is k-hop from the initiator, shown as Fig. \ref{fig1} as an example. Obviously, the revenue should be calculated on a game-by-game manner instead of user-by-user. For a single game, the more people involved, the greater the revenue. But the relation between revenue and the number of participants is not linear, we construct a valid quantitative model to compute the revenue from an initiated game. 
\begin{rem}
Even that our model is called as Collaborate game model, but here, "game" is abstract concept, which is not just limited to the game. It can be extended to other multiplayer activity circumstances.
\end{rem}
\noindent
For example, president election, in order to win a high level of support, each camp will promote and advocate their own presidential candidate by inviting some people to support them. If an invited person expresses his/her support for a presidential candidate on social media, it is possible to attract more people in his/her social circle to support this candidate together. It can be regarded as an activity, and the revenue can be gained from that defined as before.

In this paper, we propose the problem of Revenue Maximization under k-hop Collaborate Game (RMKCG): By inviting limited users to initiate a k-hop Collaborate Game, it aims to maximize the expected total revenue to the company or organization. The objective function of RMKCG problem is based on three parameters: budget $b$, at most $b$ invitations can be sent; acceptance vector $\vec{\theta}$: the probability for each user to accept the invitation from company; revenue vector $\vec{R}$: the revenue generated by those participant with different level. Before sending an invitation to a potential initiator, we do not know whether he/she will accept it and how many users around him/her will follow and join together. Thus, we adopt adaptive strategy, where we need to observe the change of state of both users and networks after sending an invitation, then decide who will be the next one. The adaptive RMKCG problem is called A-RMKCG. Golovin et al. \cite{golovin2011adaptive} formulated the theory of adaptive optimization, and they proposed two concepts: adaptive monotonicity and adaptive submodularity, if satisfying that, a $(1-1/e)$-approximation can be obtained by adaptive greedy strategy. However, the challenge is that we prove our A-RMKCG problem is not adaptive submodular, and the computation of marginal revenue is very inefficient. They will be solved partially in this paper. Our contributions are summarized as follows:

\begin{enumerate}
	\item Collaborate Game Model is a totally new model, which is a generalization to a class of real problems. Based on that, we propose RMKCG problem and its adaptive version A-RMKCG firstly.
	\item We prove A-RMKCG problem is NP-hard, adaptive monotone, but not adaptive submodular.
	\item We propose a new method to compute the marginal gain, which overcomes the difficulty of computation, improve efficiency and reduce running time greatly. Then, we prove A-RMKCG problem is adaptive submodular under some special cases, where we can obtain $(1-1/e)$-approximation by Adaptive-Invitation policy.
	\item Our proposed algorithms are evaluated on real world social networks, which verify the effectiveness and correctness of them.
\end{enumerate}

\textbf{Organization:} Sec. \uppercase\expandafter{\romannumeral2} reviews the related survey. Sec. \uppercase\expandafter{\romannumeral3} descirbes model, background knowledges and problem formulation. The solution for A-RMKCG are presented in Sec. \uppercase\expandafter{\romannumeral4}. Sec. \uppercase\expandafter{\romannumeral5} is the theoretical analysis for A-RMKCG. Sec. \uppercase\expandafter{\romannumeral6} discusses experiment and Sec. \uppercase\expandafter{\romannumeral7} is conclusion.

\begin{figure}[!t]
	\centering
	\includegraphics[width=3.5in]{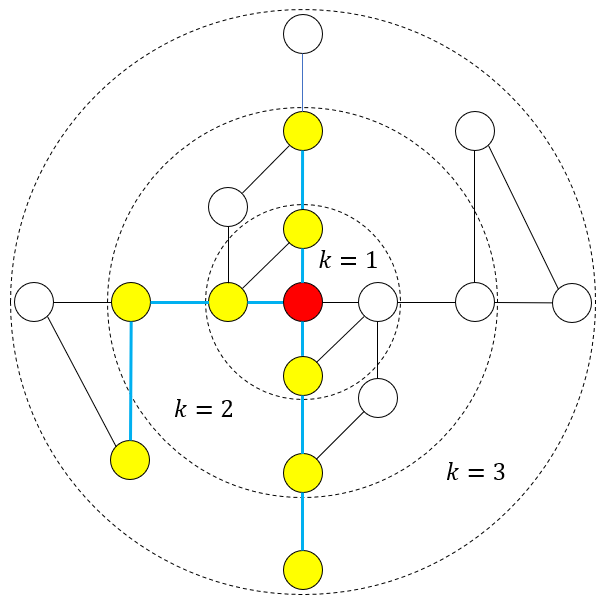}
	\caption{An example that shows a single activity: Here, the red node is initiator, and the yellow node is those users that participate the activity launched by red node.}
	\label{fig1}
\end{figure}

\section{Related Work}
Domingos and Richardson \cite{domingos2001mining} \cite{richardson2002mining} was the first to study viral marketing and the value of customers in social networks. Kempe et al. \cite{kempe2003maximizing} generalized viral marketing to influence maximization. It can be considered as a combinatorial optimization probelm, and they proposed a greedy algorithm implemented by Monte-Carlo simulations. Afterward, there have been many varients of influence maximization, among which, profit maximization is related to us. \cite{lu2012profit} \cite{tang2016profit} proposed the problem how to select the most influential seed nodes that can maximize the profit. Lu et al. \cite{lu2012profit} combined prices and valuation by extending the Linear Threshold model, then they used a heuristic unbudgeted greedy framework to solve this problem. Tang et al. \cite{tang2016profit} provided a strong approximation guarantee with the help of the methods in unconstrained submodular maximization. Other researches on pricing strategies can refer to \cite{arthur2009pricing} \cite{zhou2015bilevel} \cite{lu2016pricing} \cite{yang2016continuous}.

As we know, the approximation ratio of greedy algorithm is $(1-1/e)$, proved by Nemhauser et al. \cite{nemhauser1978analysis}, when the objective function is monotone and submodular. Golovin et al. \cite{golovin2011adaptive} extended this work to adaptive version and obtained the same approximation ratio when the objective function is adaptive monotone and adaptive submodular under the full-adoption feedback model. Under the myopic feedback model, Peng et al. \cite{NIPS2019_8795} had proven that adaptive greedy algorithm admits a constant approximation and Salha et al. \cite{salha2018adaptive} proposed the myopic adaptive greedy policy that is guaranteed to provide a $(1-1/e)$-approximation under a variant of the independent cascade model. Applied it to social networks, adaptive influence maximization with partial feedback model was proposed and studied \cite{ijcai2017-546} \cite{tang2020influence}. Tong et al. \cite{tong2019adaptive} provided a systematic studies on the adaptive influence maximization problem, where the objective function is not adaptive submodular, and they introduced the concept of regret ratio in designing the seeding strategy. Smith et al. \cite{smith2018approximately} introduced the adaptive primal curvature to obtain an approximation ratio for non-submodular cases. When the objective function is not adapative monotone, but adaptive submodular, Gotoves et al. \cite{gotovos2015non} extended random greedy algorithm to adaptive version and obtained a $(1/e)$-approximation. Other researches on the application of adaptive strategy can refer to \cite{gabillon2013adaptive} \cite{fern2017adaptive} \cite{yuan2017adaptive} \cite{han2018efficient}.

\section{Problem and Preliminaries}
In this section, we discuss about the formulation of our problem and preliminary knowledges, mainly including network model and notation, to the rest of paper.
\subsection{Collaborate Game}
To understand our game model, we need to know how it works firstly. To generalize our problem, we imagine a social relationship network $G$, from Facebook or WeChat, that a game company want to promote their new collaborate game based on this targeted network $G$. The targeted network $G$ can be defined arbitrarily by game company, for example, sub-network of children, sub-network of students or all users in a network according to the targeted users.

Let us describe how a collaborate game works. First, the game company needs to acquire the targeted network $G$ through some social platforms, such as WeChat, the most popular social software in China. It contains some necessary users' information we can exploit, for example, age, gender and career. More importantly, most social media softwares have a real-name certification system, which helps us extract the targeted network according to the properties of game. For Tencent, it is more convenient, because Tencent is both a social software company and a game company. Thus for the game promotion department, such scenes are staged every day. When the company invites a user $u$ to play this game, if he/she accepts it, he/she will do as follows:
\begin{enumerate}
	\item Initiate this game as initiator.
	\item Invite his/her friends to participate in this game.
	\item If his/her friends are willing to join the game, they will invite their friends to participate in this game again as before.
\end{enumerate}
Even though this game is a multiplayer collaborate game, it does not mean they can invite their friends to participate without limit. Thus, we have a k-hop assumption, where the invitations can only last for $k$ rounds. Explaining in detail, we call initiator as 0-hop participant. The friends of 0-hop participant that accept 0-hop participant's invitation are called 1-hop participants. The friends of 1-hop participants that accept 1-hop participants' invitation are called 2-hop participants, repeated until k-hop participants. In other words, the distance of all participants for this game from initiator cannot be larger than k-hop in targeted network $G$, which limit the scope of a single game.

However, for the game company, it is unrealistic to send too many invitations to get more initiators because they may not be able to get complete network information and determine who will be a potential initiator. Then, for the promotion, giving too many invitations is sometimes counterproductive, because this will make their game very cheap and lack competitiveness. Relied on the above analysis, our problem is who should be the next initiator that are the most beneficial to maximize the total revenue.

\subsection{Network Descriptions}
A targeted network can be given by a undirected graph $G=(V,E)$ where $V=\{v_1,v_2,...,v_n\}$ is the set of $n$ users. $E=\{e_1,e_2,...,e_m\}$ is the set of $m$ undirected edges, where $e=\{u,v\}\in E$ indicates that there exists friendship between user $u$ and user $v$. The node set and edge set for graph $G$ can be referred as $V(G)$ and $E(G)$ respectively. For an edge $e=\{u,v\}$, $u$ is a friend of $v$, naturally, $v$ is a friend of $u$ as well. We use $N(v)$ to denote the set of friends of node $v$. Here, we define a probability $p_e\in[0,1]$ for each edge $e\in E(G)$, which represents the degree of intimacy between user $u$ and user $v$ in our model. In other words, when user $u$ participates in this game, whether user $v$ is willing to participate in this game with user $u$ if $u$ invites $v$. Obviously, the more intimate the two friends are, the more likely they are willing to play games together.

For the company, they cannot know exactly how much the degree of intimacy between two users is, thus, the network information is incomplete. From the perspective of data mining, they can predict the probability between two users by learning their communication log to judge the degree of closeness between them. This is beyond the scope of this paper and we will not discuss here. Once a user $u$ accepts the invitation from game company as an initiator, the statuses of those edges whose distance is k-hop from user $u$ are partially known. According to the previous observation, the company can make a decision about which potential initiator it the best next, which is the reason why it is called an adaptive strategy.

Then, there is a natural question whether the user would accept this invitation as an initiator when he/she receives it from the game company. We can define an acceptance probability $\theta_u\in[0,1]$ for each user $u\in V(G)$, which describes the extent to which users are interested in launching this game when he/she recieves the invitation. From here, an acceptance vector $\vec{\theta}=(\theta_1,\theta_2,...,\theta_n)$ is formulated to give the acceptance probability for each user. It is complicated because different users have a divergent tendency to this game. For example, a game enthusiast may be more inclined to accept the invitation, or a user with many friends may be happy to accept ths invitation due to the fact that this game requires multiple people to participate in. Thus, it is flexible about how to define acceptance vector. In this paper, we assume $\theta_u$ is uniformly distributed in interval $[0,1]$.

\subsection{Problem Definition}
The company that promotes a game is naturally hoping to get revenue from it, and how can we characterize the revenue from a game. Obviously, the more people involved, the greater the revenue they can get. However, it is not enough to measure the income just by the number of people playing the game. This is a multiplayer game, the actual number of initiated games is much less than the number of all participants. Thus, we need to consider comprehensively both the number of initiated games and the number of participants in each game. Usually, for each initiatied game, we believe the marginal revenue is diminishing gradually, in other words, the increasing rate in earnings is gradually slowing down as the number of participants increases. Let us see an example:
\begin{exmp}
	A user $u$ accepts the invitation from the game company to be an initiator, we assume that this company can obtain $5$ units revenue from him. User $u$ invites his/her friends, soon afterward, the company can obtain $4$ units revenue from each that accepts the invitation from user $u$, namely, $4$ units/1-hop participant. Then, we have $3$ units/2-hop participant, $2$ units/3-hop participant, and so on until k-hop participant.
\end{exmp}
\noindent
This assumption is valid because we believe the promotion effect brought by the people who first join this game is larger than the people who join the game afterward. Those who join this game first need more people to get involved since they need collaboration, however, people who join this game later do not have this requirement.

In order to quantify the revenue the company obtains, we can define a revenue vector $\vec{R}=(R_0,R_1,R_2,...,R_k)$, where $R_i\in\mathbb{Z}^+$ and $i\in[k]=\{0,1,2,...,k\}$. Here, $R_i$ denotes the revenue the company can gain from a i-hop participant, thus, we have $R_0\geq R_1\geq R_2\geq...\geq R_k$ according to our previous assumption. In the targeted network, the number of initiators (0-hop participants) is clear, but other participants are likely to be ambiguous. For example, a user $u$ is a i-hop participant to one initiator and j-hop participant to another initiator, where $i<j$. At this moment, we consider user $u$ will choose to be a i-hop participant to the company. We call it as "Tendency Assumption". Based on the above model, for a game company, they aim to give a certain number of invitations to initial users such that maximizing the expected revenue. The problem of Revenue Maximization under k-hop Collaborate Game (RMKCG) is defined as follows:
\begin{pro}[RMKCG]
	Given a targeted network $G=(V,E)$, an acceptance vector $\vec{\theta}$, a budget $b$ and a revenue vector $\vec{R}$, we aim to find a subset $D\subseteq V(G)$ and $|D|\leq b$ such that the expected total revenue the company gains can be maximized when inviting those users in $D$ to be initiators.
\end{pro}
\noindent
For adaptive strategy, the RMKCG problem can be transformed to find a policy $\pi$, where the company will send an invitation to user $u\in D$ step by step. When the user $u$ accepts to be an initiator, the partial statuses of edges that are within k-hop from $u$ are known, thus, the network information can be updated. A valid policy $\pi$ can at most send $b$ invitations.

The adaptive strategy can be considered as a stochastic process, whose stochasticity is mainly from the following two aspects: For user $u$, the probability ($\theta_u$) that he/she accepts the invitation from the company as an initiator; For user $v$, the probability ($p_{uv})$ that he/she joins the game when his/her friend $u$ invites him/her to play together. Now, we can define the state of targeted network. Given a targeted network $G=(V,E)$, for each user $u\in V(G)$, the state of $u$ can be denoted by $X_u\in\{0,1\}\cup\{?\}$, where $X_u=1$ means user $u$ accepts to be an initiator because of the company's invitation and $X_u=0$ means user $u$ rejects to be an initiator. $X_u=?$ means user $u$ is unknown, who did not receive an invitation from the company. The states of all users are $?$ at the beginning. Similarly, for each edge $e=\{u,v\}\in E(G)$, the state of $e$ can be denoted by $Y_e=\{0,1\}\cup\{?\}$, where $Y_e=0$ indicates user $u$ (resp. $v$) did not accept the invitation from user $v$ (resp. $u$) and $Y_e=1$ indicates user $u$ (resp. $v$) is willing to play the game with user $v$ (resp. $u$). Once determined, it cannot be changed. $Y_e=?$ indicates there is no invitation that happens between user $u$ and user $v$. The states of all edges are $?$ at the beginning.

After defining the states of users and edges, we have a function mapping likes
\begin{equation*}
\phi=\{X_u\}_{u\in V(G)}\cup\{Y_e\}_{e\in E(G)}\rightarrow\{0,1\}^{V(G)}\cup\{0,1\}^{E(G)}
\end{equation*} 
called a realization (the states of all items in $V(G)$ and $E(G)$). Thus, we say that $\phi(u)$ is the state of user $u\in V(G)$ and $\phi(e)$ is the state of edge $e\in E(G)$ under realization $\phi$. We use $\Phi$ to denote a random realization and $\Pr(\phi)=\Pr[\Phi=\phi]$ as the probability distribution over all realizations. Besides, each realization should be consistent. Here, each user can only be one of state in $X_u\in\{0,1\}\cup\{?\}$, and each edge be one of state in $\{0,1\}\cup\{?\}$ identically. Considering the adaptive version of RMKCG problem (A-RMKCG Problem), the game company does as follows:
\begin{enumerate}
	\item Initialize $D=\emptyset$
	\item Send an invitation to user $u\in V(G)\backslash D$, observe $\Phi(u)$
	\item Update $D=D\cup\{u\}$
	\item If $\Phi(u)=0$, go to step (6)
	\item Update the states of edges whose distance are less than k-hop from user $u$, and record these update as a reference for next decision
	\item If $|D|<k$, go back to step (1); Otherwise, stop
\end{enumerate}
\noindent
From step (5), we know that this process satisfies the definition of full-adoption feedback model, where the company is able to see the entire diffusion process of each initiator they select and then, decide who will be the next one.

\begin{figure}[!t]
	\centering
	\includegraphics[width=3.5in]{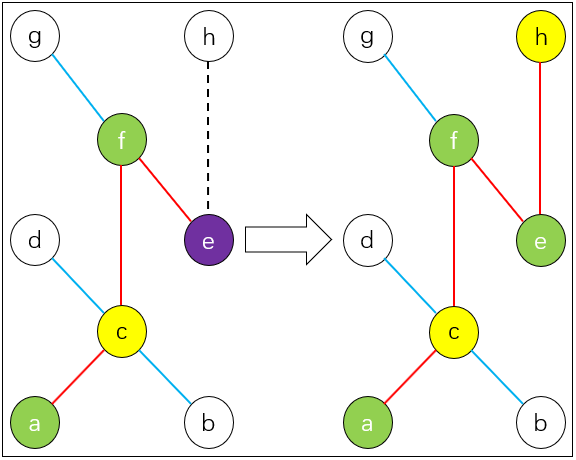}
	\caption{An example that shows the adaptive process: Here, we assume $k=2$, where the state of yellow nodes are $1$, other nodes are $?$; Green and purple nodes are 1-hop and 2-hop participants; The state of red, blue and dotted line are $1$, $0$ and $?$. First, we invite node $c$, the states shown as left part; Then, we invite node $h$, the states shown as right part. We can see that  node $e$ is changed from 2-hop to 1-hop participant because of node $h$.}
	\label{fig2}
\end{figure}

Let us see Fig. \ref{fig2} as a concrete example. We define $H(\pi,\phi)$ as the set of all users who are invited by the game company according to strategy $\pi$ under the realization $\phi$. After each invitation, shown as above process, the states of partial edges can be updated, our observation so far can be represented as a partial realization $\psi$, which is a function of observed objects to their states. Then, $dom(\psi)$ is referred to as the domain of $\psi$, namely, observed users and edges in $\psi$. A partial realization $\psi$ is consistent with a realization $\phi$ if they are equal everywhere in the domain of $\psi$, denoted by $\phi\sim\psi$. If $\psi$ and $\psi'$ are both consistent with some $\phi$, and $dom(\psi)\subseteq dom(\psi')$, we say $\psi$ is a subrealization of $\psi'$, denoted by $\psi\subseteq\psi'$. \cite{golovin2011adaptive} Besides, we denote by $dx(\psi)$ the observed users in the domain of $\psi$, and $dy(\psi)$ the observed edges in the domain of $\psi$.

Let $\pi$ be an adaptive invitation policy of the company. The total revenue gained according to policy $\pi$ under the realization $\phi$ can be defined as follows:
\begin{equation}
f(H(\pi,\phi),\phi)=\sum_{i\in[k]}\sum_{u\in D_i(\pi,\phi)}R_i
\end{equation}
where $[k]=\{0,1,2,...,k\}$ and $D_i(\pi,\phi)$ is the set that contains all i-hop participants to the company according to strategy $\pi$ under realization $\phi$, thus, we have
\begin{flalign}
&D_0(\pi,\phi)=\{u|u\in H(\pi,\phi),\phi(u)=1\}\\
&D_i(\pi,\phi)=\nonumber\\
&\left\{u|\exists v\in D_{i-1}(\pi,\phi),\phi(\{u,v\})=1\right\}\backslash\bigcup_{j=0}^{i-1}D_j(\pi,\phi)
\end{flalign}
Finally, we can evaluate the performance of a policy $\pi$ by its expected revenue, and we have
\begin{equation}
f_{avg}(\pi)=\mathbb{E}_\Phi[f(H(\pi,\Phi),\Phi)]
\end{equation}
where the expectation is taken with respect to $\Pr(\Phi=\phi)$. The goal of A-RMKCG problem is to find a policy $\pi^*$ such that $\pi^*\in\arg\max_{\pi}f_{avg}(\pi)$ subject to $|H(\pi,\phi)|\leq b$ for all realization $\phi$.

\section{Algorithm}
In this section, we propose our algorithm to solve A-RMKCG problem, Adaptive-Invitation Algorithm, and introduce how to compute marginal gain efficiently.

\subsection{Adaptive-Invitation Algorithm}
According to the description of A-RMKCG problem, inspired by adaptive greedy policy proposed by \cite{golovin2011adaptive}, Adaptive-Invitation Algorithm is proposed, which can be divided into two steps generally: In the first step, we send an invitation to user $u$ that can obtain the most increment of expected revenue according to partial realization $\psi$. This step can be generalized to Conditional Expected Marginal Revenue, which is slightly different from \cite{golovin2011adaptive},

\begin{defn}[Conditional Expected Marginal Revenue]
	Given a partial realization $\psi$ and an user $u$, the conditional expected marginal revenue of $u$ conditioned on have observed $\psi$ is
	\begin{equation*}
	\Delta(u|\psi)=\mathbb{E}\left[f(dx(\psi)\cup\{u\},\Phi)-f(dx(\psi),\Phi)|\Phi\sim\psi\right]
	\end{equation*}
\end{defn}
\noindent
In our A-RMKCG problem, $\Delta(u|\psi)$ is the expected total revenue based on previous invited users and observed edges, and the expectation is taken over all realization that are consistent with current partial realization $\psi$. After inviting user $u$ as an initiator, in the second step, we need to observe the state of $u$. If $u$ accepts this invitation, it updates current observation, namely, updates the states of edges whose distance within k-hop from $u$; If $u$ rejects this invitation, back to the first step to invite next user. The Adaptive-Invitation Algorithm is shown in Algorithm \ref{a1}. They will be executed iteratively until the number of invitation by company is larger than $b$.
\begin{exmp}
We consider a graph with five users and four edges $\{v_1,v_2\}$, $\{v_2,v_3\}$, $\{v_3,v_4\}$ and $\{v_4,v_5\}$. Suppose that $p_e=0.5$ for each edge $e$ and $\theta_v=1$ for each node $v$, budget $b=2$ and hop $k=2$. According to Algorithm \ref{a1}, the company will first invite $v_3$ undoubtedly. Suppose that $v_2$ accepts, but $v_4$ rejects their friend $v_1$'s invitations, then $v_1$ rejects his/her friend $v_2$'s invitations. Now, the states (partial realization $\psi_1$) can be shown as follows:
\begin{figure}[H]
	\centering
	\includegraphics[width=3.5in]{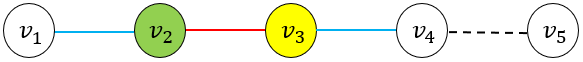}
\end{figure}
\noindent
What will happen next? We should compute their marginal gains according to $\psi_1$. We have $\Delta(v_1|\psi_1)=R_0$, $\Delta(v_2|\psi_1)=R_0-R_1$, $\Delta(v_4|\psi_1)=R_0+R_1/2$ and $\Delta(v_5|\psi_1)=R_0+R_1/2$. Thus, the company will invite $v_4$ or $v_5$ as the second initiator. Suppose that the company invite $v_4$, and $v_5$ accepted his/her friend $v_4$'s invitation. Now, the states (partial realization $\psi_2$) can be shown as follows:
\begin{figure}[H]
	\centering
	\includegraphics[width=3.5in]{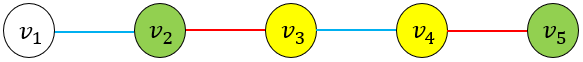}
\end{figure}
\noindent
Here, we have $H(\pi_a,\psi_2)=\{v_3,v_4\}$ and $f(H(\pi_a,\psi_2),\psi_2)=2(R_0+R_1)$ eventually.
\end{exmp}

\begin{algorithm}[!t]
	\caption{\textbf{Adaptive-Invitation $(G,f,\vec{\theta},\vec{R},b,k)$}}\label{a1}
	\begin{algorithmic}[1]
		\STATE Initialize: $H\leftarrow\emptyset$, $\psi\leftarrow\emptyset$
		\FOR {$i=1$ to $b$}
		\FOR {user $u\in V(G)\backslash H$}
		\STATE $\Delta(u|\psi)\leftarrow$ Computed by Monte-Carlo simulations or simplified method in Algorithm \ref{a2}
		\ENDFOR
		\STATE Select $u_i\in\arg\max_u\Delta(u|\psi)$
		\STATE $H\leftarrow H\cup\{u_i\}$
		\IF {$u^*$ accepts the invitation}
		\STATE Update $\psi$, the states of the edges whose distance are less than $k$-hop from $u$
		\ENDIF
		\ENDFOR
		\RETURN $f(H,\psi)$
	\end{algorithmic}
\end{algorithm}

Actually, in line 3 to 4 of Algorithm \ref{a1}, we do not need to compute $\Delta(u|\psi)$ for each user $u\in V(G)\backslash H$ again at each iteration, becuase the value of $\Delta(u|\psi)$ depends only on the states of edges within k-hop from $u$. If they do not change, the value of $\Delta(u|\psi)$ is consistent with last iteration natrually. In iteration $i$, provided that the distance between user $u_{i-1}$ and $u$ is larger than $2k$, the states of edges within $k$-hop from $u$ maintain consistency with iteration $i-1$, thus, we do not need to compute $\Delta(u|\psi)$ again. Therefore, it is necessary to maintain a record for the value of each $\Delta(\cdot|\psi)$. In iteration $i$, we check whether the distance between $u$ and $u_{i-1}$ is less or equal to $2k$, if yes, update the value of $\Delta(u|\psi)$. It improves the efficiency of Algorithm \ref{a1} partially.
\begin{rem}
Unfortunately, we know that the social networks usually show a property of small world, which implies that when $k=3$, most of the users are within 6-hop from users $u_{i-1}$. Hence, only a small fraction of nodes do not require an update when $k\geq 3$. Of course, if we can know the structure information of the targeted network in advance, it will help us decide whether we need to use the above method, because finding the distance from user $u_{i-1}$ is time-consuming.
\end{rem}

\subsection{Gain Computation}
In line 4 of Algorithm \ref{a1}, we need to compute the value of $\Delta(u|\psi)$ given current realization $\psi$, which is deterministic. Due to the fact that the $\Delta(u|\psi)$ is an expectation with respect to $\Phi\sim\psi$, the general method is to use Monte-Carlo simulations, where we run this diffusion process many times and then take the average of them. But it is extremely time-comsuming in order to get an accurate value. Thus, in this subsection, we talk about how to compute $\Delta(u|\psi)$ efficiently, the main idea is shown in Algorithm \ref{a2}.

\begin{algorithm}[!t]
	\caption{\textbf{Compute $(G,f,\vec{\theta},\vec{R},k,\psi,u)$}}\label{a2}
	\begin{algorithmic}[1]
		\STATE Initialize: list $L=[]$, set $S=\emptyset$
		\STATE $L=L+\{u:0\}$
		\STATE $S=S\cup\{u\}$
		\STATE $\Delta(u|\psi)= 0$
		\FOR {$i=1$ to $k$}
		\STATE Initialize: map $V_i=\{\}$
		\FOR {$v$ in $L[i-1]$}
		\FOR {$v'$ in $N(v)$}
		\IF {$v'\notin S$} 		
		\STATE $V_i[v']=0$
		\STATE $S=S\cup\{v'\}$
		\ENDIF
		\ENDFOR
		\ENDFOR
		\STATE $L=L+V_i$
		\ENDFOR
		\STATE $L[0][u]=\theta_u$
		\STATE $\Delta(u|\psi)+=\theta_u\cdot(\vec{R}[0]-\vec{R}[h(u|\psi)])$
		\FOR {$i=1$ to $k$}
		\FOR {$v$ in $L[i]$}
		\STATE Initialize: $t=1$
		\FOR {$v'$ in $N(v)$}
		\IF {$v'$ in $L[i-1]$}
		\IF {$\{v',v\}\in dom(\psi)$ and $\psi(\{v',v\})==1$}
		\STATE $t=t*(1-L[i-1][v'])$
		\ELSIF {$\{v',v\}\notin dom(\psi)$}
		\STATE $t=t*(1-p_{vv'}\cdot L[i-1][v'])$
		\ENDIF
		\ENDIF
		\ENDFOR
		\STATE $L[i][v]=1-t$
		\ENDFOR
		\FOR {$v$ in $L[i]$}
		\IF {$i<h(v|\psi)$}
		\STATE $\Delta(u|\psi)+=L[i][v]\cdot(\vec{R}[i]-\vec{R}[h(v|\psi)])$
		\ENDIF
		\ENDFOR
		\ENDFOR
		\RETURN $\Delta(u|\psi)$
	\end{algorithmic}
\end{algorithm}

For each user $u$, we assume there is a map containing a user $\{u:a\}$, where $a\in[0,1]$ is a probability. First, we initialize a list $L$, whose elements are maps, whose elements are users and their probabilities. For example, we set $k=2$, and we have $L=[\{u_1:a_1\},\{u_2:a_2,u_3:a_3\},\{u_4:a_4,u_5:a_5\}]$. We say $L[0]=\{u_1:a_1\}$ is initiator and its acceptance probability, $L[1]=\{u_2:a_2,u_3:a_3\}$ are 1-hop participants and their probabilities that they participate in the game initiated by user $u_1$ and $L[2]$ are 2-hop participants and their probabilities to that. In line $5$ to line $16$ of Algorithm \ref{a2}, we find all possible participants within k-hop from the initiator and put them in the correct position in list $L$, namely, 0-hop participant in $L[0]$, 1-hop participants in $L[1]$ and so on. The initial probabilities are set to $0$. Then, we define function $h(u|\psi)$, that is
\begin{equation*}
h(u|\psi)=
\begin{cases}
i&u \text{ is }\text{i-hop participant under }\psi\\
\infty&u \text{ is not participant}
\end{cases}
\end{equation*}
Here, we define $\vec{R}[\infty]=0$. In line 17 to 18 of Algorithm \ref{a2}, we set the acceptance probability to user $u$ as an initiator, and then, update $\Delta(u|\psi)$. Beginning from line $19$, for all possible i-hop participants, we set their participation probability from line 20 to 32. To understand the idea of this part, let us see following example:
\begin{exmp}
	We assume user $\{u:a\}\in L[2]$, $a=0$, is a possible $2$-hop participant to initiator $u$, and there are three users $\{x_1:a_1\}$, $\{x_2:a_2\}$, $\{x_3:a_3\}\in L[1]$ that existing edges $\{x_1,u\}$, $\{x_2,u\}$, $\{x_3,u\}\in E(G)$. If $\{x_1,u\}\in dom(\psi)$, $\psi(\{x_1,u\})=1$; $\{x_2,u\}\in dom(\psi)$, $\psi(\{x_2,u\})=0$ and $\{x_3,u\}\notin dom(\psi)$, we can update $a$ as $a=1-(1-a_1)(1-0)(1-p_{ux_3}\cdot a_3)$.
\end{exmp}
\noindent
Then, from line 33 to 37, it aims to compute the revenue gain for all possible i-hop participants. For user $v$, $\{v:a\}\in L[i]$, if $i<h(v|\psi)$, we can obtain the gain $a\cdot(\vec{R}[i]-\vec{R}[h(v|\psi)])$ according to the assumption in Section $3.3$, if $i\geq h(v|\psi)$, no gain. The algorithm is completed.

It is worth noting that this computational method is not absolutely accurate because of the complexity of A-RMKCG problem. Consider the targeted network $G=\{e_1=\{v_1,v_2\},e_2=\{v_2,v_3\},e_3=\{v_1,v_3\}\}$, $\theta_{v_1}=1$, $k=2$, the expected revenue from user $v_2$ when the company sends an inviation to $v_1$, that is,
\begin{flalign}
&=(\Pr[\Phi(\{e_1\})=1])\vec{R}[1]\nonumber\\
&+(\Pr[\Phi(\{e_1\})=0]\Pr[\Phi(\{e_2\})=1]\Pr[\Phi(\{e_3\})=1])\vec{R}[2]\nonumber
\end{flalign}
We neglect the second term in Algorithm \ref{a2}. In the most cases, the value of the second term is much less than that of the first term. If necessary to calculate $\Delta(u|\psi)$ accurately, we are still required to use Monte-Carlo simulations. Obviously, the more times the simulation is performed, the more accurate the target value and the longer the running time is. Algorithm \ref{a2} improves operational efficiency greatly under the premise of ensuring accuracy.

\subsection{Time Complexity}
Considering the line 4 of Algorithm \ref{a1}, to compute $\Delta(u|\psi)$ based on Monte-Carlo simulations, the running time is bounded by $O(mr)$. Thus, the total running time is $O(bnmr)$ where $r$ is the number of Monte-Carlo simulations. With the help of the simplified gain computation, shown in Algorithm \ref{a2}, let us first analyze the running time to compute $\Delta(u|\psi)$ by Algorithm \ref{a2}. Here, the running time is bound by $O(n+m)$. Thus, the total running time is $O(bn(n+m))$. As we know, to estimate $\Delta(u|\psi)$ accurately, the value of $r$ is very large. The time complexity of Adaptive-Invitation algorithm is improved tremendously if computing $\Delta(u|\psi)$ by Algorithm \ref{a2} instead of Monte-Carlo simulations.

\section{Theoretical Analysis}
In this section, we talk about the hardness and adaptive submodularity of A-RMKCG problem, and get an approximation ratio of Algorithm \ref{a1}.
\subsection{Hardness}
In order to show the hardness of A-RMKCG problem, we can start from a classical NP-hard problem, Maximum Coverage (MC) problem, and reduce MC to A-RMKCG problem in polynomial time. The decision version of MC can be defined as follows:
\begin{defn}[MC]
	Given an integer $b$, a collection of sets $S=\{S_1,S_2,...,S_m\}$ and an integer $Q$, we ask whether it exists a subcollection $S'\subseteq S$ such that $|S'|\leq b$ and the number of covered elements $|\bigcup_{S_i\in S'}S_i|\geq Q$.
\end{defn}

\begin{thm}
	A-RMKCG problem is NP-hard, because its special case can be reduced to MC in polynomial time.
\end{thm}
\begin{proof}
	The decision version of A-RMKCG problem: Given an interger $b$, a targeted network $G=(V,E)$ and an integer $Q'$, we ask whether it exists a policy $\pi$ and $|H(\pi,\phi)|\leq b$ for all $\phi$ such that $f_{avg}(\pi)\geq Q'$. To get a special case, we make some assumptions about the original problem: We set $k=1$, initiator only send invitation to his/her directed friends; the revenue from each participant is $1$, revenue vector $\vec{R}=(1,1)$; the acceptance probability for each user is $1$, acceptance vector $\vec{\theta}=(1,1,1,...,1)$; and the probability $p_e$ for each edge is $1$, the realization is unique.
	
	Let us construct the equivalent relation between MC and A-RMKCG. Given an instance of MC, we can define an instance of A-RMKCG as: $W=\bigcup_{S_i\in S}S_i$, for each node $w_j\in W$, we create a node $v_j$ in the instance of A-RMKCG. For each set $S_i\in S$, we create a node $v_i'$ in the instance of A-RMKCG. Thus, we have $V(G)=\{v_1,v_2,...,v_{|W|}\}\cup\{v_1',v_2',...,v_{|S|}'\}$. For each node $v_j\in S_i$, we create an undirected edge between $v_j$ and $v_i'$, thus, $E(G)=\bigcup_{i=1}^{|S|}\left(\bigcup_{j=1}^{|S_i|}\{v_i',v_j\}\right)$. The construction can be done in polynomial time and shown in Fig. \ref{fig3}. Based on the above assumptions, all users in contructed graph $G$ would accepted the invitations if they receive it. Let $Q'=Q+b$, we have:
	
	MC$\Rightarrow$A-RMKCG: Given an instance of MC, $S'\subseteq S$, $|S'|\leq b$ and $|\bigcup_{S_i\in S'}S_i|\geq Q$, as an instance of A-RMKCG, we make the policy $\pi$ invite the users in $D=\{v_i'|S_i\in S'\}$ in any order. We can know that $|H(\pi,\phi)|\leq b$ for all $\phi$, $\phi$ is unique, and $f_{avg}(\pi)\geq Q+b=Q'$.
	
	A-RMKCG$\Rightarrow$MC: Given an instance of A-RMKCG, a policy $\pi$ such that $|H(\pi,\phi)|\leq b$ and $f_{avg}(\pi)\geq Q+b$, we have following observation. First, in the constructed network $G$, if $\pi$ selects a node $v_j$, $v_j\in\{v_1,v_2,...,v_{|W|}\}$, and it does not select a $v_i'$, $v_i'\in \{v_1',v_2',...,v_{|S|}'\}$, that connects to $v_j$, then we assume there is another policy $\kappa$ that selects one of $v_i'$ that connects to $v_j$ instead of $v_j$. Obviously, we have $f_{avg}(\kappa)\geq f_{avg}(\pi)$ because the number of neighbors of $v_i'$ is greater or at least equal to that of $v_j$, thus, we can make more users participate in this game when inviting $v_i'$ instead of $v_j$. As an instance of MC, we select the subcollection $S'=\{S_i|v_i'\in H(\kappa,\phi)\}$. We can know that $|S'|\leq b$ and $|\bigcup_{S_i\in S'}S_i|\geq f_{avg}(\kappa)-b\geq f_{avg}(\pi)-b\geq Q$.
	
	Therefore, the simplest special case of A-RMKCG can be reduced to MC problem, which means that it has at least the same hardness as MC, A-RMKCG problem is NP-hard. The proof is completed.
\end{proof}

\begin{figure}[!t]
	\centering
	\includegraphics[width=3.5in]{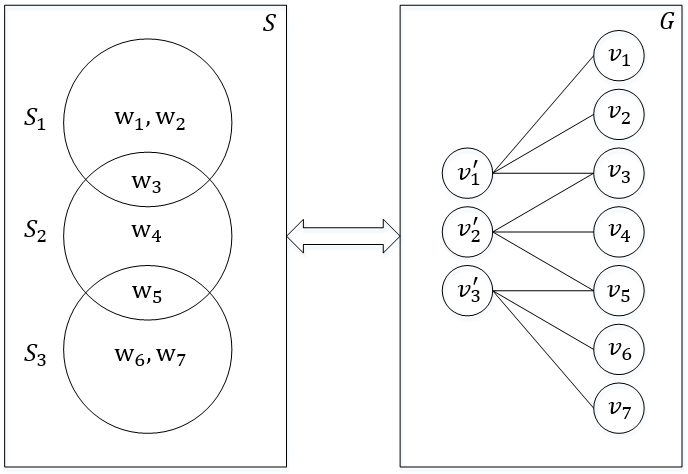}
	\caption{A sketch that shows an example: reduction between MC and A-RMKCG problem.}
	\label{fig3}
\end{figure}

\subsection{Approximation Performance}
In \cite{golovin2011adaptive}, Golovin et al. proposed two important concepts: Adaptive Monitonicity and Adaptive Submodularity, and prove we can obtain a $(1-1/e)$-approximate solution by adaptive greedy strategy if the adaptive optimization problem satisfies these two properties. Our A-RMKCG problem is an instance of adaptive optimization problem, thus, supposing our objective function, Equation (4), is adaptive monotone and adaptive submodular, we can obtain a similar result. Based on \cite{golovin2011adaptive}, these two concepts are as follows:
\begin{defn}[Adaptive Monotonicity]
	A function $f$ is adaptive monotone with respect to distribution $\Pr(\phi)$ if the conditional expected marginal revenue of any user $u$ is nonnegative, i.e., for all $\psi$ with $\Pr[\Phi\sim\psi]>0$ and all $u\in V(G)\backslash dx(\psi)$, we have
	\begin{equation}
	\Delta(u|\psi)\geq 0
	\end{equation}
\end{defn}
\begin{defn}[Adaptive Submodularity]
	A function $f$ is adaptive submodular with respect to distribution $\Pr(\phi)$ if the conditional expected marginal revenue of any user $u$ does not increase as more states of users and edges are observed, i.e., for all $\psi$ and $\psi'$ such that $\psi$ is a subrealization of $\psi'$ ($\psi\subseteq \psi'$) and all $u\in V(G)\backslash dx(\psi')$, we have
	\begin{equation}
	\Delta(u|\psi)\geq\Delta(u|\psi')
	\end{equation}
\end{defn}
Then, we can obtain our theoretical results, which is described as follows:
\begin{lem}
	The objective function $f$ of the A-RMKCG problem is adaptive monotone.
\end{lem}
\begin{proof}
	Considering a fixed observation $\psi$, for a user $u\in V(G)\backslash dx(\psi)$, if sending an invitation to $u$, there are two situations that can happen: accept or reject. If user $u$ accept to be an initiator, the value of marginal gain is at least equal to $\vec{R}[0]-\vec{R}[h(u|\psi)]$; otherwise, there is no marginal gain. the marginal gain is nonnegative in both cases. Thus, under any realization $\phi$, we have $f(dx(\psi)\cup\{u\},\phi)\geq f(dx(\psi),\phi)$. The expected marginal gain $\Delta(u|\psi)$ is the linear combination of each realization, so $\Delta(u|\psi)\geq 0$ as well.
\end{proof}

\begin{lem}
	The objective function $f$ of the A-RMKCG problem is not adaptive submodular.
\end{lem}
\begin{proof}
	Considering an example that targeted graph is $G=\{\{v_0,v_1\},\{v_1,v_2\},\{v_1,v_3\}\}$, $k=2$,  $p_{v_0v_1}=0.1$, $p_{v_1v_2}=0.1$, $p_{v_1v_3}=0.1$, we have two partial realization $\psi_1=\emptyset$ and $\psi_2=\{v_0,\{v_0,v_1\},\{v_1,v_2\},\{v_1,v_3\}\}$. $\psi_1$ means they do not receive any invitation from company, and $\psi_2$ means $v_0$ accepts to be an initiator and observes that $\{v_0,v_1\}$, $\{v_1,v_2\}$ and $\{v_1,v_3\}$ exist. Clearly, $\psi_1\subseteq\psi_2$. Relied on Definition 1, $\Delta(v_1|\psi_1)=0.5\cdot(R_0+3\times0.1\times R_1)$ because of $\mathbb{E}[\theta_{v_1}]=0.5$. However, $\Delta(v_1|\psi_2)=0.5\cdot[(R_0-R_1)+2\times(R_1-R_2)]$. Here, assume $\vec{R}=(5,3,1)$, we have $\Delta(v_1|\psi_1)=2.95$ and $\Delta(v_1|\psi_2)=3$. Thus, $\Delta(v_3|\psi_1)<\Delta(v_3|\psi_2)$, $f$ is not adaptive submodular.
\end{proof}
\noindent
Even though that, in some special cases, it is adaptive submodular unexpectedly. Let us see
\begin{rem}
	Given two partial realization that $\psi\subseteq \psi'$, we have $h(v|\psi)\geq h(v|\psi')$ for each user $v$ due to the fact that tendency assumption shown as before.
\end{rem}

\begin{lem}
	If our A-RMKCG problem conforms one of following two special cases:
	\begin{enumerate}
		\item $k\leq 1$
		\item For all $\{u,v\}\in E(G)$, $p_{uv}=1$
	\end{enumerate}
	The objective function $f$ is adaptive submodular.
\end{lem}
\begin{proof}
	In order to prove adaptive submodularity, we need to show, for any partial realization $\psi$, $\psi'$ such that $\psi\subseteq \psi'$, we have $\Delta(u|\psi)\geq\Delta(u|\psi')$. Similar to the proof of adaptive monotonicity, we consider two fix observation $\psi$, $\psi'$ such that $\psi\subseteq \psi'$ and a user $u\in V(G)\backslash dx(\psi')$. Given an observation $\psi$, we define the marginal gain under realization $\phi\sim\psi$ as $\Delta[(u|\psi)|\phi]$:
	\begin{equation}
	\Delta(u|\psi,\phi\sim\psi)=f(dx(\psi)\cup\{u\},\phi)-f(dx(\psi),\phi)
	\end{equation}
	Assuming that there are two realizations such that $\phi\sim\psi$ and $\phi'\sim\psi'$, we have $\phi(v)=\phi'(v)$ for all $v\notin dx(\psi')$ and $\phi(\{u,v\})=\phi'(\{u,v\})$ for all $\{u,v\}\notin dy(\psi')$, thus they have the same part $\alpha=\psi\cup(\phi'\backslash \psi')$. Now we can prove
	\begin{equation}
	\Delta(u|\psi,\phi\sim\psi)\geq\Delta(u|\psi',\phi'\sim\psi')
	\end{equation}
	for all $u\in dx(\psi')$. We will discuss the above two cases seperately as follows:
	\begin{enumerate}
		\item For the first case, when $k=0$, Equation (8) is established obvious because of Remark $1$. When $k=1$, if $u$ accepts the invitation, we have $\vec{R}[0]-\vec{R}[h(u|\psi)]\geq\vec{R}[0]-\vec{R}[h(u|\psi')]$, the gain of $u$ under $\psi$ is eqaul or larger than that under $\psi'$. For each user $v\in N(u)$, if $\{u,v\}\in dom(\alpha)$, the gain of $v$ under $\psi$ is eqaul or larger than that under $\psi'$ regardless $\alpha(\{u,v\})=1$ or $0$. If $\{u,v\}\in dom(\psi'\backslash\psi)$, when $\psi'(\{u,v\})=0$, there is no gain of $v$ under $\psi'$; when $\psi'(\{u,v\})=1$, $\psi'(v)=0$ definitely, there is no gain of $v$ under $\psi'$. Thus, whatever $\phi(\{u,v\})$ is, the gain of $v$ under $\psi$ is eqaul or larger than that under $\psi'$. If $u$ rejects the invitation, there is no gain both $\psi$ and $\psi'$. Integrating all situations, Equation (8) is established.
		\item For the second case, for all $\{u,v\}\in E(G)$, $p_{uv}=1$, the state of edges under realization $\phi$ and $\phi'$ are identical. $\psi\subseteq \psi'$ means that $dx(\psi)\subseteq dx(\psi')$, based on Remark 1, it is esay to get Equation (8).
	\end{enumerate}
	
	According to Eqaution (7), Equation (8) and Definition 1, we have as follows: $\Delta(u|\phi)=$
	\begin{flalign}
	&=\sum_{\phi\sim\psi}\Pr[\phi|\phi\sim\psi]\Delta(u|\psi,\phi\sim\psi)\nonumber\\
	&=\sum_{\phi'\sim\psi'}\Pr[\phi'|\phi'\sim\psi']\sum_{\phi\sim\alpha}\Pr[\phi|\phi\sim\alpha]\Delta(u|\psi,\phi\sim\psi)\nonumber
	\end{flalign}
	Since Equation (8) and $\sum_{\phi\sim\alpha}\Pr[\phi|\phi\sim\alpha]=1$,
	\begin{flalign}
	&\geq\sum_{\phi'\sim\psi'}\Pr[\phi'|\phi'\sim\psi']\sum_{\phi\sim\alpha}\Pr[\phi|\phi\sim\alpha]\Delta(u|\psi',\phi'\sim\psi')\nonumber\\
	&\geq\sum_{\phi'\sim\psi'}\Pr[\phi'|\phi'\sim\psi']\Delta(u|\psi',\phi'\sim\psi')\sum_{\phi\sim\alpha}\Pr[\phi|\phi\sim\alpha]\nonumber\\
	&=\sum_{\phi'\sim\psi'}\Pr[\phi'|\phi'\sim\psi']\Delta(u|\psi',\phi'\sim\psi')\nonumber\\
	&=\Delta(u|\phi')\nonumber
	\end{flalign}
	Therefore, $\Delta(u|\psi)\geq\Delta(u|\psi')$, the proof of adaptive submodularity is completed.
\end{proof}
\begin{rem}
	Even though the A-RMKCG problem is not adaptive submodular, we can know that the objective function $f$ is getting close to be adaptive submodular, when k becomes smaller or the probabilities of edges are approaching to $1$. Shown as the proof of Lemma 2, it is a typical non-submodular case, the likelihood such case happens is reduced gradually when $k$ becomes smaller or $p_{uv}\rightarrow 1$ for each $\{u,v\}\in E(G)$.
\end{rem}
\begin{thm}
	The adaptive policy $\pi_a$, given by Adaptive-Invitation algorithm, for our A-RMKCG problem is a $(1-1/e)$-approximate solution when $k\leq 1$ or $p_{uv}=1$ for each $\{u,v\}\in E(G)$. Hence, we have
	\begin{equation}
	f_{avg}(\pi_a)\geq(1-1/e)\cdot f_{avg}(\pi^*)
	\end{equation}
\end{thm}
\begin{proof}
	Based on Lemma 1 and Lemma $3$, the objective function $f$ is adaptive monotone and submodular, adaptive greedy policy is a $(1-1/e)$-approximation according to the conclusion of \cite{golovin2011adaptive}.
\end{proof}

\section{Experiment}
In this section, we need to validate the effectiveness and correctness of our proposed adaptive polices on several real social networks. The datasets in our experiments are from networkrepository.com \cite{nr}, which is an website of network repository. Three datasets with different size are used in our experiments. The dataset-1 is a co-authorship network, where each edge is a co-authorship among scientists in network theory and experiments. The dataset-2 is a Wiki network, which is a who-votes-on-whom network collected from Wikipedia. The dataset-3 is a social friendship network extracted from Facebook consisting of people with edges representing friendship ties. The statistics information of the three datasets is represented in table \ref{table1}.

\begin{table}[h]
	\renewcommand{\arraystretch}{1.3}
	\caption{The statistics of the graphs used in our experiments \cite{nr}}
	\label{table1}
	\centering
	\begin{tabular}{|c|c|c|c|c|}
		\hline
		\bfseries Dataset & \bfseries n & \bfseries m & \bfseries Type & \bfseries Avg. Degree\\
		\hline
		Dataset-1 & 0.4K & 1.01K & undirected & 4.00\\
		\hline
		Dataset-2 & 1.0K & 3.15K & undirected & 6.00\\
		\hline
		Dataset-3 & 3.0K & 65.2K & undirected & 48.0\\
		\hline
	\end{tabular}
\end{table}

\subsection{Experimental Settings}
As mentioned earlier, our proposed algorithms associated with A-RMKCG are based on the following parameters: Hop number $k$ (at most $k$-hop participants follow an initiator), Acceptance vector $\vec{\theta}$, Revenue vector $\vec{R}$, budget $b$ and edge probability. For each $e\in E(G)$, we set $p_e=0.5$ and for each user $u$, $\theta_u$ is uniformly distributed in $[0,1]$. There are two experiments we perform in this part: Algorithm \ref{a2} performance and Algorithm \ref{a1} performance. For the performance of Algorithm \ref{a2}, we have shown that the value of $\Delta(u|\psi)$ computed by Algorithm \ref{a2} is not precise completely when $k\geq 2$, and accurate value can be obtained by Monte-Carlo simulation. In this experiment, we aim to evaluate the effectiveness and efficiency of our Algorithm \ref{a2}, compared with Monte-Carlo simulation. We run the Monte-Carlo simulation $100$ times and take the average of them to compute the value of $\Delta(u|\psi)$.

For the performance of Algorithm \ref{a1}, in this experiment, we test our Algorithm \ref{a1} with some common heuristic algorithms and compare the results of performance. It aims to evaluate the effectiveness of our adaptive invitation strategy. This experiment can be divided into three parts: $k=1$, $k=2$ and $k=3$, and their revenue vector is set as $(8,6)$, $(8,6,4)$ and $(8,6,4,2)$, which means that the revenue of 0-hop participant is $8$ units, 1-hop is $6$ units, 2-hop is $4$ units and 3-hop is $2$ units. Our proposed algorithms are compared with some baseline algorithms:
\begin{enumerate}
	\item MaxDegree: Invite the user with maximum degree at each step within budget $b$.  
	\item Random: Invite a user randomly from $V(G)$ at each step within budget $b$. 
	\item MaxProb: Invite the user with maximum acceptance probability at each step within budget $b$.
	\item MaxDegreeProb: Invite the user with maximum product of degree multiplying acceptance probability at each step within budget $b$.
\end{enumerate}
We use python to test each algorithms. The simulation is run on a Windows machine with a 3.40GHz, 4 core Intel CPU and 16GB RAM.

\begin{figure}[!t]
	\centering
	\includegraphics[width=3.5in]{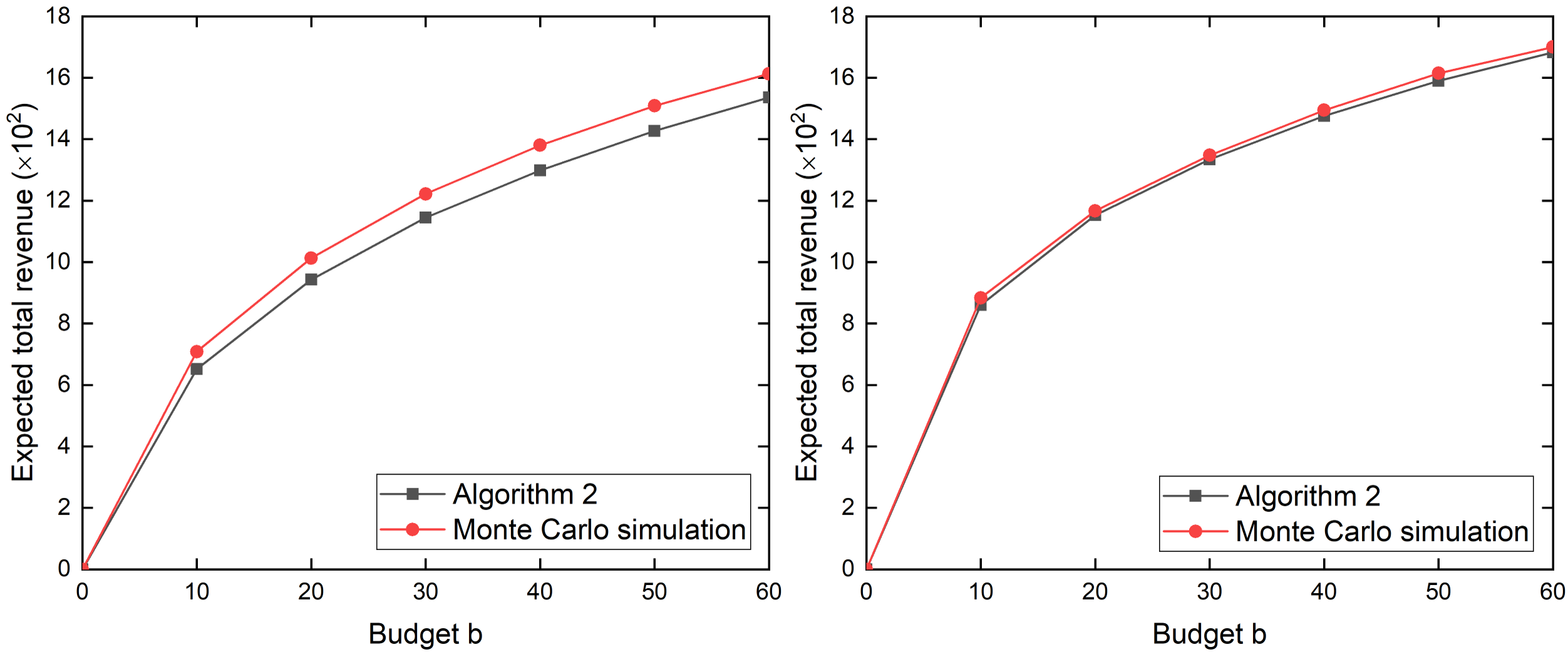}
	\caption{The performance of Adaptive-Invitation changes over budget $b$ under the dataset-1. To compute $\Delta(u|\psi)$, one is by Algorithm \ref{a2}, another one is by Monte-Carlo simulations. The left figure is under $k=2$ and right figure is under $k=3$.}
	\label{fig4}
\end{figure}

\begin{figure}[!t]
	\centering
	\includegraphics[width=3.5in]{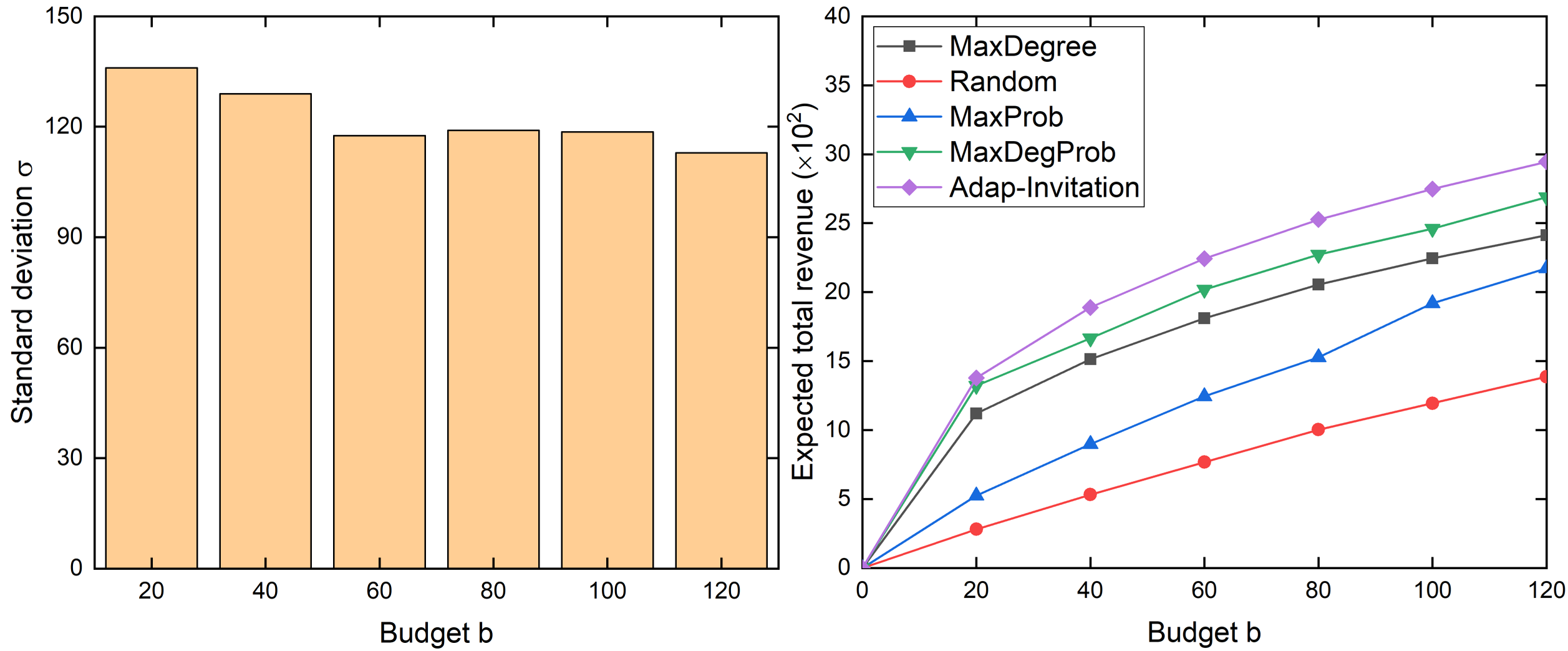}
	\\(a) $k=1$
	\\${}$
	\includegraphics[width=3.5in]{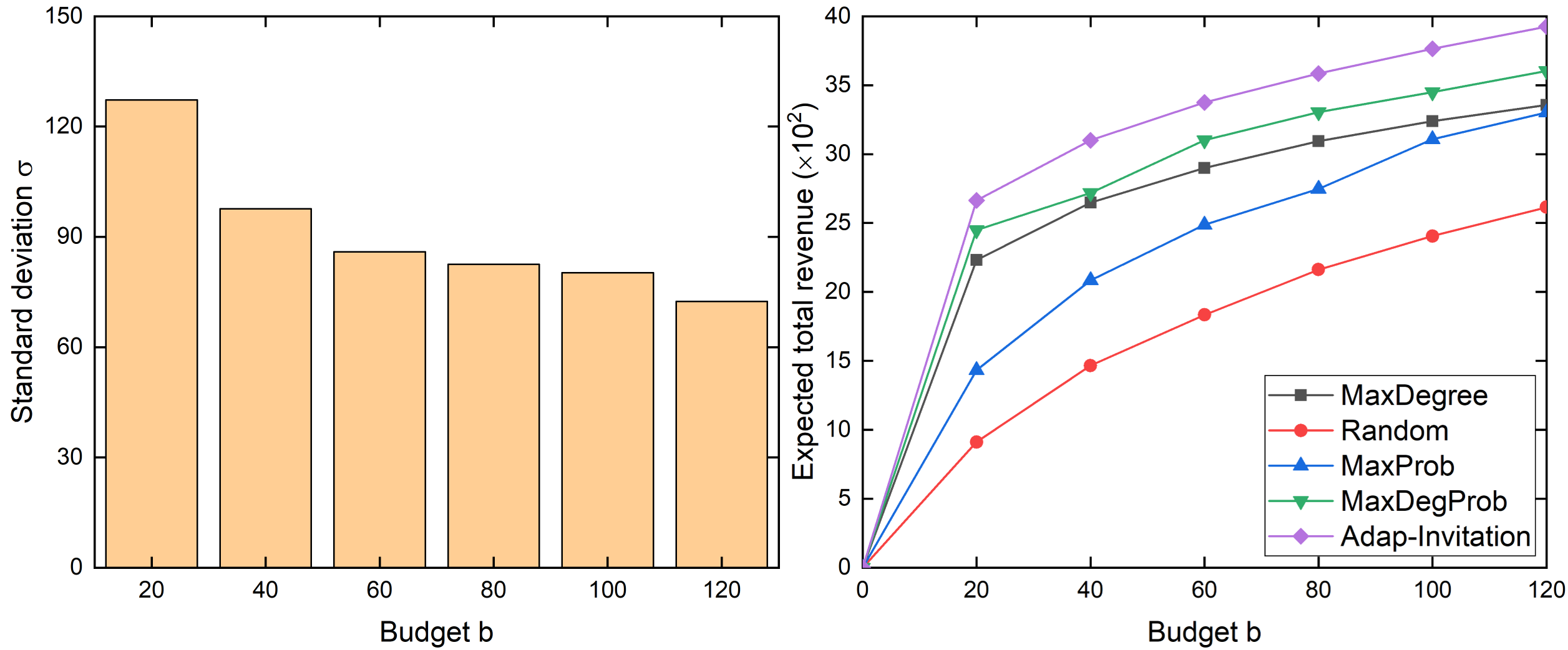}
	\\(b) $k=2$
	\\${}$
	\includegraphics[width=3.5in]{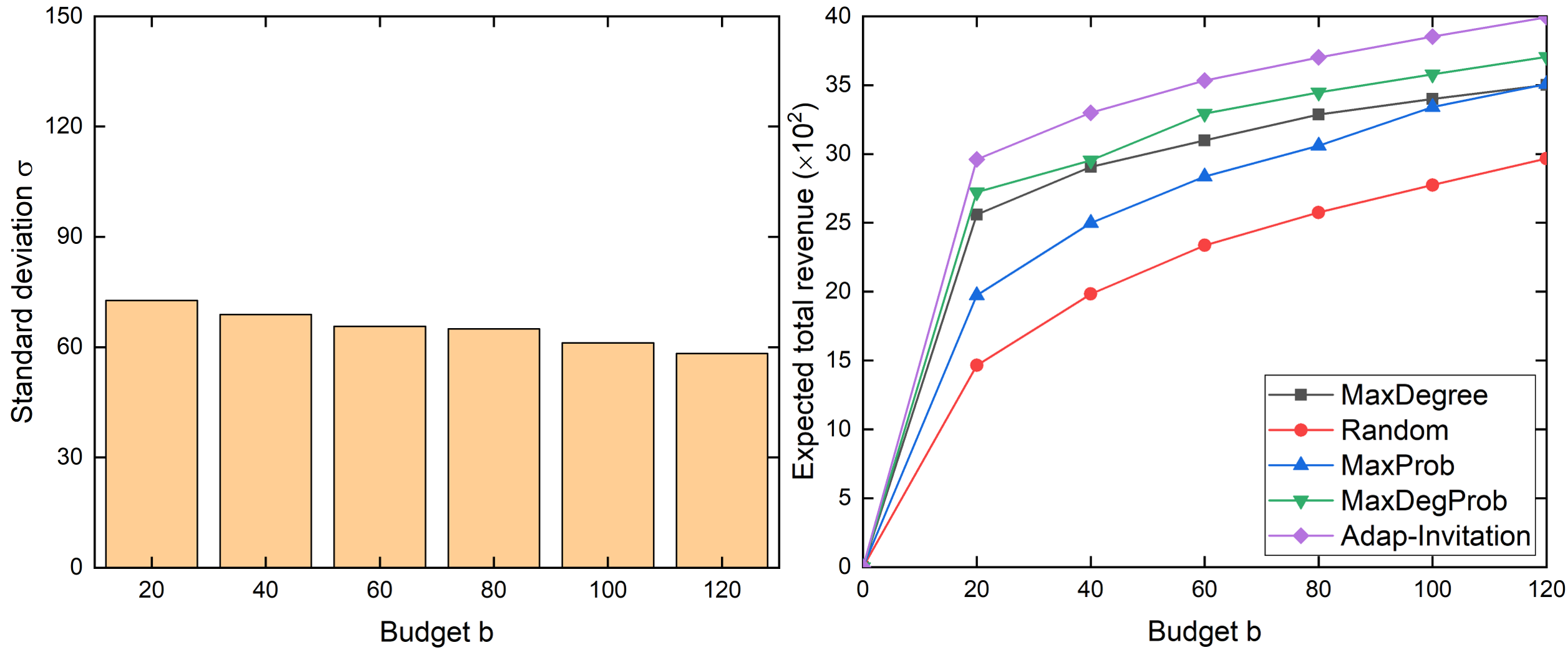}
	\\(b) $k=3$
	\caption{The performance comparison changes between Adaptive-Invitation and other heuristic algorithms over budget $b$ under dataset-2. Left column is the stardard deviation of Adaptive-Invitation; Right column is the performance comparison.}
	\label{fig5}
\end{figure}

\begin{figure}[!t]
	\centering
	\includegraphics[width=3.5in]{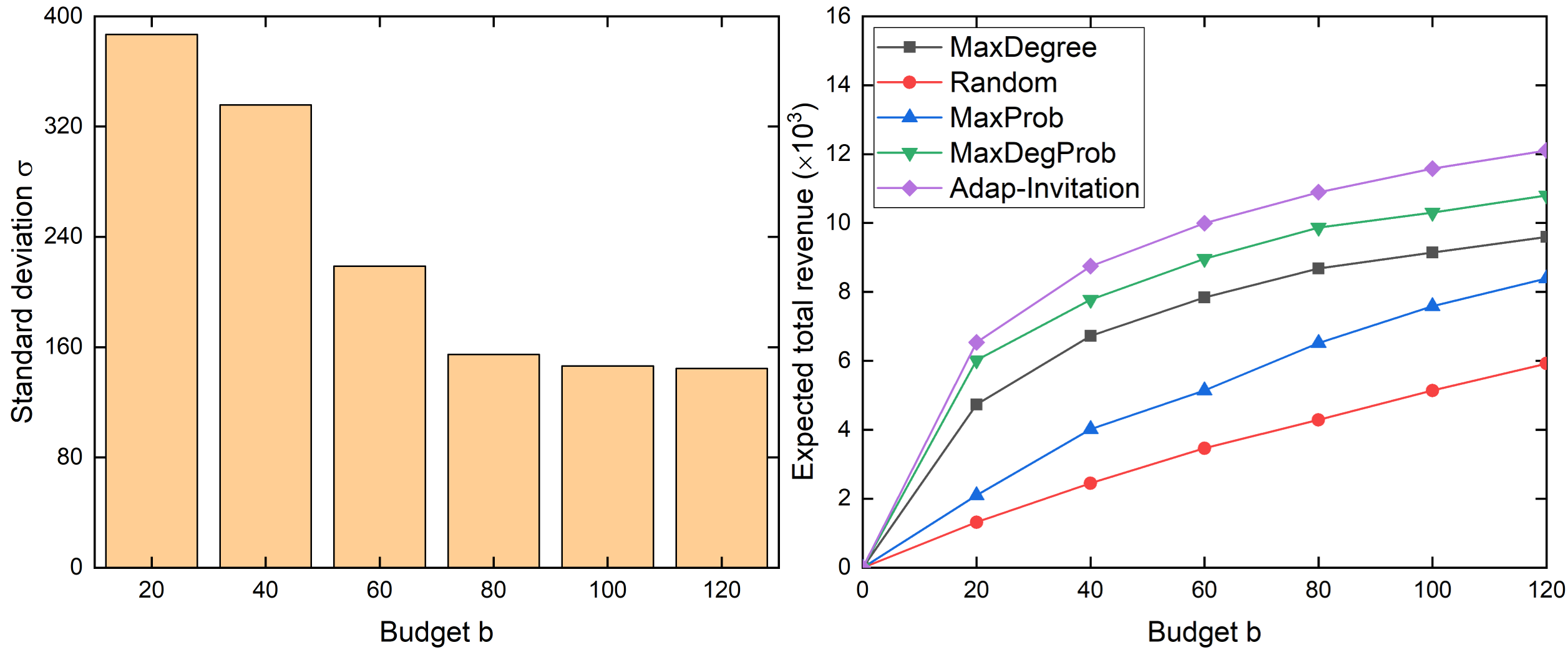}
	\\(a) $k=1$
	\\${}$
	\includegraphics[width=3.5in]{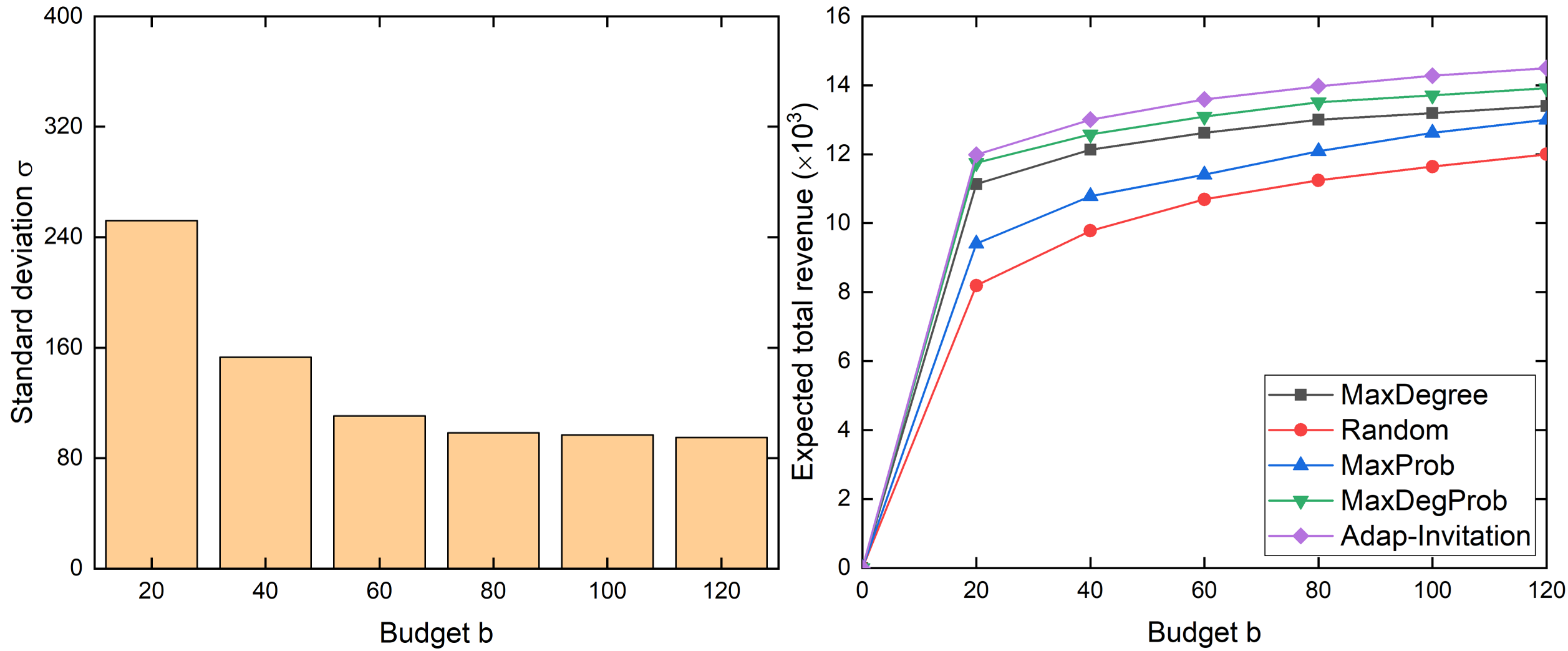}
	\\(b) $k=2$
	\\${}$
	\includegraphics[width=3.5in]{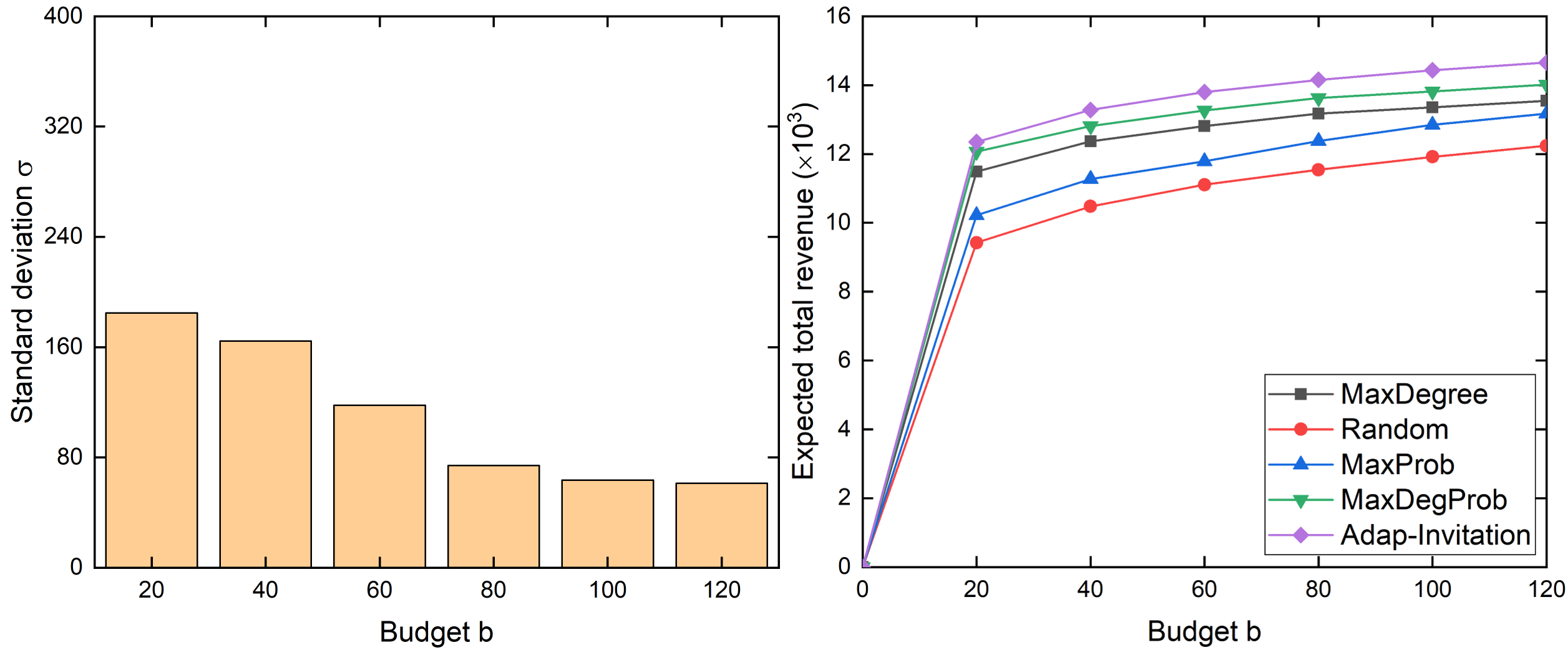}
	\\(b) $k=3$
	\caption{The performance comparison changes between Adaptive-Invitation and other heuristic algorithms over budget $b$ under dataset-3. Left column is the stardard deviation of Adaptive-Invitation; Right column is the performance comparison.}
	\label{fig6}
\end{figure}

\begin{table}[h]
	\renewcommand{\arraystretch}{1.3}
	\caption{The running time comparison}
	\label{table2}
	\centering
	\begin{tabular}{|c|c|c|c|c|}
		\hline
	    \multirow{2}*{} & \multicolumn{2}{|c|}{k=2} & \multicolumn{2}{|c|}{k=3}\\
	    \cline{2-5}
	    & Alg. 2 & M-C & Alg. 2 & M-C\\
		\hline
		b=10 & 0.339s & 544.731s & 1.663s & 858.971s\\
		\hline
		b=20 & 0.578s & 836.811s & 3.077s & 1432.74s\\
		\hline
		b=30 & 0.755s & 1052.06s & 4.244s & 1947.38s\\
		\hline
		b=40 & 0.892s & 1242.95s & 5.318s & 2410.27s\\
		\hline
		b=50 & 1.011s & 1417.15s & 6.309s & 2862.12s\\
		\hline
		b=60 & 1.163s & 1603.03s & 7.253s & 3251.53s\\
		\hline
	\end{tabular}
\end{table}

\subsection{Experimental Results}
In our experiments, the whole graphs are considered as the targeted networks. We run these adaptive algorithms on three datasets, and for each algorithm, we simulate $50$ times and take the average of them, besides, we record the standard deviation for Adaptive-Invitation algorithm. The analysis of these experimental results is summarized as follows.

Fig. \ref{fig4} draws the performance achieved by Adaptive-Invitation algorithm under the dataset-1, which aims to compare the performance and running time of Algorithm \ref{a1}, computing $\Delta(u|\psi)$ by Algorithm \ref{a2} or by Monte-Carlo simulations. In other words, in line 4 of Algorithm \ref{a1}, one uses Algorithm \ref{a2}, another one uses Monte-Carlo simulations. Shown as Fig. \ref{fig4}, the total revenue achieved by use of Algorithm \ref{a2} and Monte-Carlo simulations is very close. It is more apparent when $k=3$, and the result returned by these two methods is almost the same. Even though the performance of Monte-Carlo simulations is slightly better than that of Algorithm \ref{a2}, this difference is acceptable and it is related to topological structure of targeted networks. However, the running time for Algorithm \ref{a2} is much less than Monte-Carlo simulations, and obviously, the larger the graph is, the more significant this gap will be. The running time comparison between Algorithm \ref{a2} and Monte-Carlo simulations is represented in table \ref{table2}. The running speed is increased by at least a thousand times by Algorithm \ref{a2}. From this meaningful results, we can see that Algorithm \ref{a2} makes A-RMKCG problem be scalable to large real social networks.

Fig. \ref{fig5} and Fig. \ref{fig6} draw the performance comparison achieved by Adaptive-Invitation and other heuristic algorithms, and standard deviation of Adaptive-Invitation algorithm under the dataset-2 and dataset-3. From the left column of Fig. \ref{fig5} and Fig. \ref{fig6}, we can observe two features: (1) The standard deviation drops with the increase of hop $k$ under the same dataset; (2) The standard deviation drops with the increase of budget $b$ under the same dataset and hop. We attempt to give a valid explanation here. The uncertainty mainly comes from the nodes' acceptance probability and the edges' probability. When $k$ is smaller, for example, $k=1$, if a user $u$ does not accept the invitation, the revenue from all his/her neighbors is $0$ certainly unless there are other initiators existing among them. Conversely, when $k$ is larger, for example, $k=3$, his/her neighbors are possible to be 1-hop or 2-hop participants even if there is no initiator, which lead to the gap of revenue reduced. Thus, larger $k$ leads to smaller standard deviation. Then, with budget $b$ increasing, marginal gain is decremented generally, and the gap brought by the difference of the number of initiators is reduced. Thus, larger $b$ leads to smaller standard deviation as well. 

From the right column of Fig. \ref{fig5} and Fig. \ref{fig6}, the expected total revenue returned by Adaptive-Invitation is better than other policies under any datasets and hop number, so its performance is the best. Among these heuristic policies, MaxDegreeProb has the best performance, because it considers the degree and acceptance probability comprehensizely. It proves our theoretical analysis in the last section. We have said that the objective function is not adaptive submodular when $k\geq 2$, but in this figure, it shows the characteristics of submodularity as well, which means that the degree of submodularity is related to the structure of networks and probability setting.

\section{Conclusion}
In this paper, we build a new model, Collaborate Game, to model some real scenarios that cannot be covered by existing model. Then, we propose an adaptive RMKCG problem, and prove it is NP-hard, adaptive monotone but not adaptive submodular. Even that, we show that under some special case, $k\leq1$ and $p_e=1$ for all $e\in E(G)$, the objective function is adaptive submodular, which can be solved within $(1-1/e)$-approximation. Besides, we propose an effective method to overcome the difficulty in computing marginal gain, which makes it be scalable. The good performance of our algorithms is verified by our experiments on three real network datasets. The future work can be divided into two parts: (1) Trying to find a more generalized model, and combine with the technique of data mining to optimize the parameter setting. (2) Trying to solve the general case, in other words, get the theoretical bound for adaptive non-submodular cases.

\section*{Acknowledgment}

This work is partly supported by National Science Foundation under grant 1747818.

\ifCLASSOPTIONcaptionsoff
  \newpage
\fi



%

\bibliographystyle{IEEEtran}
\bibliography{references}

%




\end{document}